\documentclass[pra,aps,nopacs,onecolumn,twoside,superscriptaddress]{revtex4}

\usepackage{multirow, makecell}
\usepackage{amsmath,amsfonts,amssymb,caption,color,epsfig,graphics,graphicx,hyperref,latexsym,mathrsfs,revsymb,theorem,url,verbatim,epstopdf,enumerate}
\usepackage{amsmath,amsfonts,amssymb,caption,hyperref,color,epsfig,graphics,graphicx,latexsym,mathrsfs,revsymb,theorem,url,verbatim,epstopdf,cleveref}
\usepackage{fontenc}
\usepackage{cases}
\usepackage[T1]{fontenc}
\hypersetup{colorlinks,linkcolor={blue},citecolor={blue},urlcolor={red}}
\usepackage{lmodern} 
\usepackage{tipa}
\usepackage{textcomp} 

\newtheorem{definition}{Definition}
\newtheorem{proposition}[definition]{Proposition}
\newtheorem{lemma}[definition]{Lemma}

\newtheorem{theorem}[definition]{Theorem}
\newtheorem{corollary}[definition]{Corollary}
\newtheorem{conjecture}[definition]{Conjecture}

\newtheorem{remark}[definition]{Remark}
\newtheorem{example}[definition]{Example}
\newtheorem{question}[definition]{Question}
\newtheorem{memo}[definition]{Memo}

\def\squareforqed{\hbox{\rlap{$\sqcap$}$\sqcup$}}
\def\qed{\ifmmode\squareforqed\else{\unskip\nobreak\hfil
		\penalty50\hskip1em\null\nobreak\hfil\squareforqed
		\parfillskip=0pt\finalhyphendemerits=0\endgraf}\fi}
\def\endenv{\ifmmode\;\else{\unskip\nobreak\hfil
		\penalty50\hskip1em\null\nobreak\hfil\;
		\parfillskip=0pt\finalhyphendemerits=0\endgraf}\fi}
\newenvironment{proof}{\noindent \textbf{{Proof.~} }}{\qed}
\def\Dbar{\leavevmode\lower.6ex\hbox to 0pt
	{\hskip-.23ex\accent"16\hss}D}
\makeatletter
\def\url@leostyle{%
	\@ifundefined{selectfont}{\def\UrlFont{\sf}}{\def\UrlFont{\small\ttfamily}}}
\makeatother
\def\bcj{\begin{conjecture}}
	\def\ecj{\end{conjecture}}
\def\bcr{\begin{corollary}}
	\def\ecr{\end{corollary}}
\def\bd{\begin{definition}}
	\def\ed{\end{definition}}
\def\bea{\begin{eqnarray}}
	\def\eea{\end{eqnarray}}
\def\bem{\begin{enumerate}}
	\def\eem{\end{enumerate}}
\def\bex{\begin{example}}
	\def\eex{\end{example}}
\def\bim{\begin{itemize}}
	\def\eim{\end{itemize}}
\def\bl{\begin{lemma}}
	\def\el{\end{lemma}}
\def\bma{\begin{bmatrix}}
	\def\ema{\end{bmatrix}}
\def\bpf{\begin{proof}}
	\def\epf{\end{proof}}
\def\bpp{\begin{proposition}}
	\def\epp{\end{proposition}}
\def\bqu{\begin{question}}
	\def\equ{\end{question}}
\def\br{\begin{remark}}
	\def\er{\end{remark}}
\def\bt{\begin{theorem}}
	\def\et{\end{theorem}}
\def\bmm{\begin{memo}}
	\def\emm{\end{memo}}

\def\btb{\begin{tabular}}
	\def\etb{\end{tabular}}

	\newcommand{\nc}{\newcommand}
	
	\def\i{\iota}

	\nc{\bbA}{\mathbb{A}} \nc{\bbB}{\mathbb{B}} \nc{\bbC}{\mathbb{C}}
	\nc{\bbD}{\mathbb{D}} \nc{\bbE}{\mathbb{E}} \nc{\bbF}{\mathbb{F}}
	\nc{\bbG}{\mathbb{G}} \nc{\bbH}{\mathbb{H}} \nc{\bbI}{\mathbb{I}}
	\nc{\bbJ}{\mathbb{J}} \nc{\bbK}{\mathbb{K}} \nc{\bbL}{\mathbb{L}}
	\nc{\bbM}{\mathbb{M}} \nc{\bbN}{\mathbb{N}} \nc{\bbO}{\mathbb{O}}
	\nc{\bbP}{\mathbb{P}} \nc{\bbQ}{\mathbb{Q}} \nc{\bbR}{\mathbb{R}}
	\nc{\bbS}{\mathbb{S}} \nc{\bbT}{\mathbb{T}} \nc{\bbU}{\mathbb{U}}
	\nc{\bbV}{\mathbb{V}} \nc{\bbW}{\mathbb{W}} \nc{\bbX}{\mathbb{X}}
	\nc{\bbZ}{\mathbb{Z}}
	
	\nc{\bA}{{\bf A}} \nc{\bB}{{\bf B}} \nc{\bC}{{\bf C}}
	\nc{\bD}{{\bf D}} \nc{\bE}{{\bf E}} \nc{\bF}{{\bf F}}
	\nc{\bG}{{\bf G}} \nc{\bH}{{\bf H}} \nc{\bI}{{\bf I}}
	\nc{\bJ}{{\bf J}} \nc{\bK}{{\bf K}} \nc{\bL}{{\bf L}}
	\nc{\bM}{{\bf M}} \nc{\bN}{{\bf N}} \nc{\bO}{{\bf O}}
	\nc{\bP}{{\bf P}} \nc{\bQ}{{\bf Q}} \nc{\bR}{{\bf R}}
	\nc{\bS}{{\bf S}} \nc{\bT}{{\bf T}} \nc{\bU}{{\bf U}}
	\nc{\bV}{{\bf V}} \nc{\bW}{{\bf W}} \nc{\bX}{{\bf X}}
	\nc{\bZ}{{\bf Z}}
	
	\nc{\as}{{\cal AS}}
	\nc{\app}{{\cal AP}}
	\nc{\ar}{{\cal AR}}
	\nc{\bp}{{\cal BP}}
	\nc{\dbp}{{\cal DBP}}
	\nc{\ew}{{\cal EW}}
	\nc{\dew}{{\cal DEW}}
	\nc{\ndew}{{\cal NDEW}}
	\nc{\conv}{{\text{Conv}}}

	\nc{\cA}{{\cal A}} \nc{\cB}{{\cal B}} \nc{\cC}{{\cal C}}
	\nc{\cD}{{\cal D}} \nc{\cE}{{\cal E}} \nc{\cF}{{\cal F}}
	\nc{\cG}{{\cal G}} \nc{\cH}{{\cal H}} \nc{\cI}{{\cal I}}
	\nc{\cJ}{{\cal J}} \nc{\cK}{{\cal K}} \nc{\cL}{{\cal L}}
	\nc{\cM}{{\cal M}} \nc{\cN}{{\cal N}} \nc{\cO}{{\cal O}}
	\nc{\cP}{{\cal P}} \nc{\cQ}{{\cal Q}} \nc{\cR}{{\cal R}}
	\nc{\cS}{{\cal S}} \nc{\cT}{{\cal T}} \nc{\cU}{{\cal U}}
	\nc{\cV}{{\cal V}} \nc{\cW}{{\cal W}} \nc{\cX}{{\cal X}}
	\nc{\cZ}{{\cal Z}}
	
	\nc{\cpp}{{\cal PP}}
	
	\nc{\hA}{{\hat{A}}} \nc{\hB}{{\hat{B}}} \nc{\hC}{{\hat{C}}}
	\nc{\hD}{{\hat{D}}} \nc{\hE}{{\hat{E}}} \nc{\hF}{{\hat{F}}}
	\nc{\hG}{{\hat{G}}} \nc{\hH}{{\hat{H}}} \nc{\hI}{{\hat{I}}}
	\nc{\hJ}{{\hat{J}}} \nc{\hK}{{\hat{K}}} \nc{\hL}{{\hat{L}}}
	\nc{\hM}{{\hat{M}}} \nc{\hN}{{\hat{N}}} \nc{\hO}{{\hat{O}}}
	\nc{\hP}{{\hat{P}}} \nc{\hR}{{\hat{R}}} \nc{\hS}{{\hat{S}}}
	\nc{\hT}{{\hat{T}}} \nc{\hU}{{\hat{U}}} \nc{\hV}{{\hat{V}}}
	\nc{\hW}{{\hat{W}}} \nc{\hX}{{\hat{X}}} \nc{\hZ}{{\hat{Z}}}
	
	\nc{\hn}{{\hat{n}}}
	
	\def\diag{\mathop{\rm diag}}

	\def\max{\mathop{\rm max}}

	\def\tr{\mathop{\rm Tr}}

	\newcommand{\ket}[1]{|#1\rangle}
	\newcommand{\proj}[1]{| #1\rangle\!\langle #1 |}

	\def\Dbar{\leavevmode\lower.6ex\hbox to 0pt
		{\hskip-.23ex\accent"16\hss}D}

\begin{document}

		\title{A counterexample to the strong spin alignment conjecture}
		
		\date{\today}

		\author{Zhiwei Song}\email[]{zhiweisong@buaa.edu.cn}
		\affiliation{LMIB(Beihang University), Ministry of education, and School of Mathematical Sciences, Beihang University, Beijing 100191, China}
		
		\author{Lin Chen}\email[]{linchen@buaa.edu.cn}
		\affiliation{LMIB(Beihang University), Ministry of education, and School of Mathematical Sciences, Beihang University, Beijing 100191, China}
		
		
	\begin{abstract}
		The spin alignment conjecture was originally formulated in connection with the additivity of coherent information for a class of quantum channels known as platypus channels. Recently, a stronger majorization-based version was proposed by M. A. Alhejji and E. Knill [Commun. Math. Phys. 405, 119, 2024], asserting that the spectrum of the alignment operator is always majorized by that of the perfectly aligned configuration. In this letter, we show that this strong spin alignment conjecture is false in general by constructing an explicit counterexample in the smallest unresolved case, namely three qubits. The example uses two-body states that are not jointly compatible with any single three-qubit global state, which naturally leads to a compatibility-constrained variant of the conjecture.
	\end{abstract}
		
		\maketitle

		\section{Introduction}
		
		The coherent transmission of quantum information through noisy channels is a central problem in quantum Shannon theory. The coherent information provides a fundamental achievable rate and, after regularization, characterizes the quantum capacity of a channel; consequently, additivity questions for coherent information are of basic importance \cite{Lloyd97,DevetakShor05}. Against this background, the spin alignment conjecture was introduced by Leditzky \emph{et al.}~\cite{platypus,nonadd} in the study of the additivity of coherent information for a class of quantum channels known as platypus channels. Let $d$ be the local dimension, let $n$ be the number of particles, let $\mu$ be a probability measure on the subsets of $[n]$, and let $Q$ be a fixed quantum state on $\bbC^d$. The original question asks whether the von Neumann entropy of
		\begin{equation}
			\sum_{I\subseteq[n]} \mu_I \, \proj{\psi_I} \otimes Q^{\otimes I^c}
		\end{equation}
		is always minimized when each freely chosen pure state $\ket{\psi_I}$ is taken to be a tensor power of a maximal eigenvector of $Q$. Physically, this corresponds to aligning all freely chosen spins with the dominant spin direction of $Q$.
		
		Alhejji and Knill reformulated the problem in terms of alignment operators and proposed a stronger version based on majorization~\cite{AK}. For Hermitian matrices $X$ and $Y$, we write $X \preceq Y$ if their eigenvalues in nonincreasing order satisfy
		\begin{equation}
			\sum_{i=1}^k \lambda_i(X) \le \sum_{i=1}^k \lambda_i(Y)
		\end{equation}
		for every $k$, with equality for the total trace. This relation implies entropy minimization for every unitarily invariant strictly concave function.
		
		\begin{conjecture}[Strong Spin Alignment Conjecture~\cite{AK}]
			\label{conj:strong_spin}
			Let $Q$ be a quantum state on $\bbC^d$, and let $\ket{q_1}$ be a maximal eigenvector of $Q$. Then for any probability measure $\mu$ on the subsets of $[n]$ and any tuple of pure states $(\ket{\psi_I})_{I\subseteq[n]}$,
			\begin{equation}
				\sum_{I\subseteq[n]} \mu_I \, \proj{\psi_I} \otimes Q^{\otimes I^c}
				\preceq
				\sum_{I\subseteq[n]} \mu_I \, \proj{q_1}^{\otimes I} \otimes Q^{\otimes I^c}.
			\end{equation}
		\end{conjecture}
		
		Several special cases were established in~\cite{AK}. Subsequently, the $n=2$ case of this conjecture was settled in~\cite{Alhejji2025}. Thus the first unresolved case is $n=3$. In the next section, we show that the strong conjecture already fails there.
		
		\section{A three-qubit counterexample}
		\label{sec:counterexample}

		For this tripartite system, let $Q = \bbI/2$ be the maximally mixed state on a single qubit, and let the probability measure $\mu$ be uniformly distributed over the 2-element subsets of $\{A,B,C\}$. The alignment operator is given by
		\begin{eqnarray}
			\label{H}
			H = \frac{1}{6} \big( \proj{\psi_{AB}} \otimes \bbI_C + \proj{\psi_{AC}} \otimes \bbI_B + \bbI_A \otimes \proj{\psi_{BC}} \big).
		\end{eqnarray}
		According to Conjecture \ref{conj:strong_spin}, the spectrum of $H$ should be majorized by the spectrum of the target operator $H_{\text{target}}$, where $\ket{q_1}$ is chosen as $\ket{0}$:
		\begin{eqnarray}
			H_{\text{target}} = \frac{1}{6} \big( \proj{0}_{A}\otimes \proj{0}_{B}\otimes \bbI_C + \proj{0}_{A}\otimes \bbI_B\otimes \proj{0}_{C} + \bbI_A \otimes \proj{0}_{B}\otimes \proj{0}_{C} \big).
		\end{eqnarray}
		The eigenvalues of $H_{\text{target}}$ sorted in descending order are exactly
		\begin{eqnarray}
			\lambda(H_{\text{target}}) = \left(\frac{1}{2}, \frac{1}{6}, \frac{1}{6}, \frac{1}{6}, 0, 0, 0, 0\right)^T.
		\end{eqnarray}
		We now construct a specific pure-state tuple $(\ket{\psi_{AB}}, \ket{\psi_{AC}}, \ket{\psi_{BC}})$ that violates the majorization condition. Let
		\begin{eqnarray}
			\ket{\psi_{AB}} &:=& \frac{1}{\sqrt{217}} \big( 6\ket{01} - 9\ket{10} - 10\ket{11} \big), \\
			\ket{\psi_{AC}} &:=& \frac{1}{\sqrt{46}} \big( -2\ket{00} - 4\ket{01} - 1\ket{10} + 5\ket{11} \big), \\
			\ket{\psi_{BC}} &:=& \frac{1}{\sqrt{53}} \big( -6\ket{00} - 1\ket{01} - 4\ket{10}\big).
		\end{eqnarray}
		Substituting into (\ref{H}), a direct calculation gives 
		\begin{eqnarray}
			\lambda_1(H)\approx 0.32276,\qquad
			\lambda_2(H)\approx 0.31510,\qquad
			\lambda_3(H)\approx 0.19654.
		\end{eqnarray}
		Consequently,
		\begin{eqnarray}
			\sum_{i=1}^3 \lambda_i(H) \approx 0.83440 > \frac{5}{6}= \sum_{i=1}^3 \lambda_i(H_{\text{target}}).
		\end{eqnarray}
		Since the cumulative sum of the first three eigenvalues of $H$ strictly exceeds the corresponding sum for $H_{\text{target}}$, the majorization relation $H \preceq H_{\text{target}}$ is false. Therefore, the strong spin alignment conjecture fails.
		
		{\bf Remark.}
		Although the majorization statement fails, the original entropy-minimization conjecture is not violated by this example. Numerically,
		\begin{eqnarray}
			S(H)\approx 2.064 \ \text{bits} > 1.792 \ \text{bits} \approx S(H_{\mathrm{target}}).
		\end{eqnarray}
		
		\section{A compatible-marginal refined conjecture}
		
		The counterexample above demonstrates that Conjecture~\ref{conj:strong_spin} can fail when the states on different subsets are chosen independently. This suggests that the failure of Conjecture~\ref{conj:strong_spin} may stem from the absence of compatibility rather than from the alignment principle itself.  Motivated by this observation and supported by numerical evidence, we formulate the following refined conjecture.
		
		\begin{conjecture}[Compatible-marginal spin alignment conjecture]
			\label{conj:compatible}
			Let $Q$ be a quantum state on $\bbC^d$, and let $\ket{q_1}$ be a maximal eigenvector of $Q$. For any $n$-partite global quantum state $\rho$ on $(\bbC^d)^{\otimes n}$, define its marginals on $I\subseteq[n]$ as
			\begin{eqnarray}
				\rho_I := \tr\nolimits_{I^c} \rho.
			\end{eqnarray}
			Then for every probability measure $\mu$ on the subsets of $[n]$, the following majorization relation holds:
			\begin{eqnarray}
				\sum_{I\subseteq[n]} \mu_I \rho_I \otimes Q^{\otimes I^c}
				\preceq
				\sum_{I\subseteq[n]} \mu_I \proj{q_1}^{\otimes I} \otimes Q^{\otimes I^c}.
			\end{eqnarray}
		\end{conjecture}

	Conjecture~\ref{conj:compatible} may be viewed as a spectral extremality statement within the quantum marginal problem, with the aligned configuration as the most spectrally concentrated compatible one. 
		As further evidence for Conjecture~\ref{conj:compatible}, we record that the compatible-marginal majorization statement is restored in the same three-qubit two-body setting where the unrestricted strong conjecture fails. Equivalently, the next proposition verifies Conjecture~\ref{conj:compatible} when $Q=\bbI/2$ and $\mu$ is supported on $AB,AC,BC$.
		
		\begin{proposition}
			\label{prop:two_body_compatible}
			Let $a,b,c\ge 0$ with $a+b+c=1$, and let $\rho_{ABC}$ be a three-qubit state. Then
			\begin{eqnarray}
				a \rho_{AB}\otimes \bbI_C+b \rho_{AC}\otimes \bbI_B+c \bbI_A\otimes \rho_{BC}
				\preceq
				a \proj{00}_{AB}\otimes \bbI_C+b \proj{00}_{AC}\otimes \bbI_B+c \bbI_A\otimes \proj{00}_{BC}.
			\end{eqnarray}
		\end{proposition}
		
		\begin{proof}
			By permuting the three subsystems if necessary, we may assume $a\ge b\ge c\ge 0$. Set
			\begin{eqnarray}
				X_1:=\rho_{AB}\otimes \bbI_C,\qquad X_2:=\rho_{AB}\otimes \bbI_C+\rho_{AC}\otimes \bbI_B,\qquad X_3:=\rho_{AB}\otimes \bbI_C+\rho_{AC}\otimes \bbI_B+\bbI_A\otimes \rho_{BC},
			\end{eqnarray}
			and
			\begin{eqnarray}
				\notag
				&&T_1:=\proj{00}_{AB}\otimes \bbI_C,\qquad T_2:=\proj{00}_{AB}\otimes \bbI_C+\proj{00}_{AC}\otimes \bbI_B, \\
				&&T_3:=\proj{00}_{AB}\otimes \bbI_C+\proj{00}_{AC}\otimes \bbI_B+\bbI_A\otimes \proj{00}_{BC}.
			\end{eqnarray}
			Then
			\begin{eqnarray}
				a \rho_{AB}\otimes \bbI_C+b \rho_{AC}\otimes \bbI_B+c \bbI_A\otimes \rho_{BC}=(a-b)X_1+(b-c)X_2+cX_3,
			\end{eqnarray}
			while the target is $(a-b)T_1+(b-c)T_2+cT_3$. 
			
		Firstly, since $\rho_{AB}\preceq \proj{00}_{AB}$ and majorization is preserved under tensoring with $\bbI_C$, we have
			\begin{eqnarray}
				\label{x1}
				X_1\preceq T_1.
			\end{eqnarray}
		
			Secondly, we write
		$\rho_{AB}=\sum_i p_i \proj{\phi_i}_{AB}, \rho_{AC}=\sum_j q_j \proj{\psi_j}_{AC}$. Then
			\begin{eqnarray}
				X_2=\sum_{i,j}p_iq_j\left(\proj{\phi_i}_{AB}\otimes \bbI_C+\proj{\psi_j}_{AC}\otimes \bbI_B\right).
			\end{eqnarray}
			For each pair $(i,j)$, the two-summand case proved in~\cite[Corollary 4.10]{AK} yields
			\begin{eqnarray}
				\proj{\phi_i}_{AB}\otimes \bbI_C+\proj{\psi_j}_{AC}\otimes \bbI_B\preceq T_2.
			\end{eqnarray}
			By convexity of the Ky Fan sums, we have
			\begin{eqnarray}
				\label{x2}
			X_2\preceq T_2.
			\end{eqnarray}
			
			We next show that $X_3\preceq T_3$. Let $K_r(M):=\sum_{i=1}^r\lambda_i(M)$ denote the $r$th Ky Fan sum. By convexity, it suffices to consider pure state $\rho=|\psi\rangle\langle\psi|$. By local unitary invariance, we may assume
			\begin{eqnarray}
				\rho_A=\diag(1-r_A,r_A),\qquad \rho_B=\diag(1-r_B,r_B),\qquad \rho_C=\diag(1-r_C,r_C),
			\end{eqnarray}
			where 
			\begin{eqnarray}
				\label{g1}
			0\le r_A,r_B,r_C\le 1/2.
			\end{eqnarray}
For pure three-qubit states these numbers satisfy the polygon inequalities~\cite{Higuchi03}
			\begin{eqnarray}
				\label{g2}
				r_A\le r_B+r_C, \qquad r_B\le r_C+r_A, \qquad r_C\le r_A+r_B.
			\end{eqnarray}
			Define
			\begin{eqnarray}
				D:=\rho_A\otimes \bbI_B\otimes \bbI_C+\bbI_A\otimes \rho_B\otimes \bbI_C+\bbI_A\otimes \bbI_B\otimes \rho_C-\bbI_{ABC},
			\end{eqnarray}
			and set $s:=r_A+r_B+r_C$ and $\delta:=(1-s)_+$. 
		In the computational basis, $D$ is diagonal with eigenvalues
		\begin{eqnarray}
			\notag
			&&2-s,\ 1-r_A-r_B+r_C,\ 1-r_A+r_B-r_C,\ 1+r_A-r_B-r_C,\ \\
   		&&-r_A+r_B+r_C,\ r_A-r_B+r_C,\ r_A+r_B-r_C,\ s-1.
		\end{eqnarray}
		Denote $h_1\ge\cdots\ge h_8$ as the eigenvalues of $D$.
		By (\ref{g1}) and (\ref{g2}), one can verify that
		$
		h_8=s-1,
		$
		and therefore $\max\{h_8,0\}=h_8+\delta$.  A direct check gives
			\begin{eqnarray}
				K_1(D)\le 1+\delta,\qquad K_2(D)\le 2+\delta,\qquad K_3(D)\le 3+\delta.
			\end{eqnarray}
			Let
			\begin{eqnarray}	
				\widetilde{\rho}:=(\sigma_y\otimes \sigma_y\otimes \sigma_y)\rho^T(\sigma_y\otimes \sigma_y\otimes \sigma_y),
	 	     \end{eqnarray}		
	 	     where $\sigma_y=\bma 0&-i\\i&0\ema$.
			Note that the identity 
		$
			\sigma_y M^T\sigma_y=(\tr M)\bbI-M
			$
			holds for arbitrary $2\times 2$ matrices.  Applying this local map to all three qubits, we can express the transformation as 
				\begin{eqnarray}
					\widetilde{\rho} = (\tr(\cdot)\bbI - \text{id})^{\otimes 3}(\rho),
				\end{eqnarray}
			where $\text{id}$ denotes the identity map.
			Expanding this tensor product and evaluating the partial traces linearly, we obtain
				\begin{eqnarray}
				\widetilde{\rho}=\bbI_{ABC}-\rho_A\otimes \bbI_B\otimes \bbI_C-\bbI_A\otimes \rho_B\otimes \bbI_C-\bbI_A\otimes \bbI_B\otimes \rho_C+\rho_{AB}\otimes \bbI_C+\rho_{AC}\otimes \bbI_B+\bbI_A\otimes \rho_{BC}-\rho.
			\end{eqnarray}
			
			Thus $X_3=D+P$, where $P:=\rho+\widetilde{\rho}\ge 0$ and $\tr P=2$. Consequently, Weyl monotonicity gives $\lambda_i(X_3)\ge h_i$ for all $i$, and also $\lambda_8(X_3)\ge \max\{h_8,0\}=h_8+\delta$. Therefore, for $r=1,2,3$,
			\begin{eqnarray}
				\sum_{i=r+1}^8\lambda_i(X_3)\ge \sum_{i=r+1}^7 h_i+h_8+\delta=\tr D-K_r(D)+\delta=4-K_r(D)+\delta,
			\end{eqnarray}
			and since $X_3\ge 0$ and $\tr X_3=6$, we obtain
			\begin{eqnarray}
				K_r(X_3)=6-\sum_{i=r+1}^8\lambda_i(X_3)\le 2+K_r(D)-\delta\le r+2,\qquad r=1,2,3.
			\end{eqnarray}
			Because $\lambda(T_3)=(3,1,1,1,0,0,0,0)$, for $r\ge 4$, we trivially have
		$
		K_r(X_3)\le \tr X_3=6=K_r(T_3).
		$
			Together with the bounds for $r=1,2,3$, 
			this proves 
			\begin{eqnarray}
				\label{x3}
			X_3\preceq T_3.
			\end{eqnarray}

		Finally, for every $r=1,\ldots,8$, Ky Fan subadditivity together with (\ref{x1}), (\ref{x2}), and (\ref{x3}) yields
		\begin{eqnarray}
			&\lefteqn{K_r\Big(a\rho_{AB}\otimes \bbI_C+b\rho_{AC}\otimes \bbI_B+c\bbI_A\otimes \rho_{BC}\Big)}\nonumber\\
			&\le&(a-b)K_r(X_1)+(b-c)K_r(X_2)+cK_r(X_3)\nonumber\\
			&\le&(a-b)K_r(T_1)+(b-c)K_r(T_2)+cK_r(T_3)\nonumber\\
			&=&K_r\big((a-b)T_1+(b-c)T_2+cT_3\big)\nonumber\\
			&=&K_r\big(a\proj{00}_{AB}\otimes \bbI_C+b\proj{00}_{AC}\otimes \bbI_B+c\bbI_A\otimes \proj{00}_{BC}\big),
		\end{eqnarray}
	where the first equality follows from a direct calculation.
	This completes the proof.
		\end{proof}

		\section{Conclusion and outlook}
		
		We have shown that the strong spin alignment conjecture fails in general by giving an explicit counterexample. 
		This motivates the compatible-marginal variant formulated in Conjecture~\ref{conj:compatible}. 
		A natural next step is to prove Conjecture~\ref{conj:compatible} in more general three-qubit settings and for arbitrary positive semidefinite $Q$. More broadly, it would be interesting to understand its connection with the quantum marginal problem and the entropy extremality questions that originally motivated spin alignment.

		\section*{ACKNOWLEDGEMENT}
	Authors were supported by the NNSF of China (Grant No. 12471427).


\begin{thebibliography}{99}
			
			\bibitem{Lloyd97}
			S.~Lloyd,
			\newblock Capacity of the noisy quantum channel,
			\newblock \emph{Phys. Rev. A} \textbf{55}, 1613--1622, 1997.
			\newblock \href{https://doi.org/10.1103/PhysRevA.55.1613}{https://doi.org/10.1103/PhysRevA.55.1613}.
			
			\bibitem{DevetakShor05}
			I.~Devetak and P.~W.~Shor,
			\newblock The capacity of a quantum channel for simultaneous transmission of classical and quantum information,
			\newblock \emph{Commun. Math. Phys.} \textbf{256}, 287--303, 2005.
			\newblock \href{https://doi.org/10.1007/s00220-005-1317-6}{https://doi.org/10.1007/s00220-005-1317-6}.
			
			\bibitem{platypus}
			F.~Leditzky, D.~Leung, V.~Siddhu, G.~Smith, and J.~A.~Smolin,
			\newblock The platypus of the quantum channel zoo,
			\newblock \emph{IEEE Trans. Inf. Theory} \textbf{69}, 3825--3849, 2023.
			\newblock \href{https://doi.org/10.1109/TIT.2023.3245985}{https://doi.org/10.1109/TIT.2023.3245985}.
			
			\bibitem{nonadd}
			F.~Leditzky, D.~Leung, V.~Siddhu, G.~Smith, and J.~A.~Smolin,
			\newblock Generic nonadditivity of quantum capacity in simple channels,
			\newblock \emph{Phys. Rev. Lett.} \textbf{130}, 200801, 2023.
			\newblock \href{https://doi.org/10.1103/PhysRevLett.130.200801}{https://doi.org/10.1103/PhysRevLett.130.200801}.
			
			
			\bibitem{AK}
			M.~A.~Alhejji and E.~Knill,
			\newblock Towards a resolution of the spin alignment problem,
			\newblock \emph{Commun. Math. Phys.} \textbf{405}, 119, 2024.
			\newblock \href{https://doi.org/10.1007/s00220-024-04980-1}{https://doi.org/10.1007/s00220-024-04980-1}.
			
			\bibitem{Alhejji2025}
			M.~A.~Alhejji,
			\newblock Refining Ky Fan's majorization relation with linear programming,
			\newblock \emph{Ann. Henri Poincar\'e} \textbf{27}, 909--932 (2026).
			\newblock \href{https://doi.org/10.1007/s00023-025-01592-w}{https://doi.org/10.1007/s00023-025-01592-w}.
			
			
			\bibitem{Higuchi03}
			A.~Higuchi, A.~Sudbery, and J.~Szulc,
			\newblock One-qubit reduced states of a pure many-qubit state: polygon inequalities,
			\newblock \emph{Phys. Rev. Lett.} \textbf{90}, 107902, 2003.
			\newblock \href{https://doi.org/10.1103/PhysRevLett.90.107902}{https://doi.org/10.1103/PhysRevLett.90.107902}.
		\end{thebibliography}
	\end{document}